\documentclass[5p,authoryear,times]{elsarticle}

\usepackage{IEEEtrantools}
\usepackage[cmex10]{amsmath}
\usepackage{amsthm}
\usepackage{amssymb}
\usepackage[tight,footnotesize]{subfigure}
\usepackage{amsfonts}
\usepackage{epstopdf}

\newtheorem{thm}{Theorem}

\newtheorem{proposition}[thm]{Proposition}
\newtheorem{corollary}[thm]{Corollary}
\newtheorem{assumption}{Assumption}
\newtheorem{definition}{Definition}
\newtheorem{rmk}{Remark}

\begin{document}

\title{From Parametric Model-based Optimization\\to robust PID Gain Scheduling}
\author[rvt]{Minh H.T. Nguyen \corref{cor1}}
\ead{tuanminh@nus.edu.sg}

\author[rvt]{K.K. Tan}
\ead{kktan@nus.edu.sg}

\cortext[cor1]{Corresponding author}
\address[rvt]{National University of Singapore, Department of Electrical and Computer Engineering,\\3 Engineering Drive 3, Singapore 117576}
%
%

\begin{abstract}
In chemical process applications, model predictive control (MPC) effectively deals with input and state constraints during transient operations. However, industrial PID controllers directly manipulates the actuators, so they play the key role in small perturbation robustness. This paper considers the problem of augmenting the commonplace PID with the constraint handling and optimization functionalities of MPC. First, we review the MPC framework, which employs a linear feedback gain in its unconstrained region. This linear gain can be any preexisting multi-loop PID design, or based on the two stabilizing PI/PID designs for multivariable systems proposed in the paper. The resulting controller is a feedforward PID mapping, a straightforward form without the need of tuning PID to fit an optimal input. The parametrized solution of MPC under constraints further leverages a familiar PID gain scheduling structure. Steady state robustness is achieved along with the PID design so that additional robustness analysis is avoided.

\end{abstract}

\maketitle

\section{Introduction}\label{intro}

Multilevel control attracts intensive research as a systematic tool for control of real plants with respect to high-level target while adhering to the local constraints \citep{Tat08Advanced}. The upper levels are usually concerned with plant-wide steady state objectives with low rate sampling. The lower levels address fast dynamic control. There is a mature trend of applying advanced optimization packages to fill the gap between these two layers. Well-known industrial examples such as AspenOne and RHMPC use MPC as the core optimizer to deal with constraints \citep{Qin03survey,Fro06Model}. MPC is a constraint-handling optimization method where the core idea is based on the receding horizon control. At each sampling time, the current plant output/state is measured, and an optimal input is derived to minimize a performance index subject to state and input constraints. This desired inputs are sent to PID controllers to directly manipulate the actuators. These PIDs must be tuned to minimize the mismatch with the updated optimal input at each sampling step. The first objective of the paper aims to bypass this two-phase complication through direct optimization of the PID gains.

Currently, there are two approaches of MPC, using either \emph{online} implementation \citep{May00Constrained} for slow processes or \emph{offline} implementation \citep{Bem02explicit} for fast processes. The former control approach solves in real time an optimization problem, thus it is more flexible to system design changes. The latter approach solves the same problem offline for all feasible states, and obtains the optimal control law in real time by searching the current state over feasible regions. This scheme, named parametric MPC, can effectively facilitate a PID gain scheduling implementation. The resulting PID controller will deal with constraints by changing gains upon the transition of active constraint regions, not at each time step. This is the second and main objective: to develop a practical implementation of parametric MPC.

The PID realization of MPC can be achieved with its robustness property intact. In fact, a great number of research methods have carefully addressed the robustness of MPC for perturbations both along the trajectory (robust performance) and at steady state (robust stability). Polytopic uncertainty model is discussed in \citet{Gri03} with LMI and in \citet{Bem03Min,Nag04Open} where min-max solutions are formed; bounded disturbances addressed by tube-based MPC is proposed in \citet{Alv08,Mar11Stochastic}. The tradeoff lies in the complexity of the solutions. In this note, we are keen on observing the robust stability provided by the simple PID form of the proposed solution.

In the literature, many finite-horizon optimal PID designs for constrained multivariable systems have been attempted to deliver a systematic PID tuning. In \citet{Mor03}, the velocity form of PID prohibited the variable gain structure, thus a fixed PID gain must be used across the prediction horizon. As shown in \citet{Camacho03Model,Aro08,Sat12} the GPC-based PID results apply to the plants approximated by a first or second order model, thus limiting their applications to multivariable plants. The solution in \citet{Di10Model} partially solves the problem, but the two controllers MPC and PID must operate in parallel. A flexible framework for optimal PIDs is still under ongoing research.

Collectively through the two mentioned objectives, this paper seeks to improve the MPC-based PID scheme to further close the gap between MPC optimization and PID controllers. In Section II, we formulate the tracking problem and analyze the controllability and observability of the augmented system. In Section III, we describes the MPC formula and shows that either a new or existing multi-loop unconstrained PID designs can be adopted into the framework. For convenience, two methods are provided to calculate the PI/PID gains at the operating point so that the closed loop system is stable. The first method applies LQR on the PI state while the latter leads to linear matrix inequalities (LMI) with the size proportional to the number of tracked outputs. Section IV applies this PID design on the piecewise affine (PWA) solution of MPC, which suggests a distributed PID gain scheduling framework to deal with constraints. Fig. \ref{fig1} shows the involved levels within the plantwide structure.

\subsection*{Notation}
The operators $\sum, \Delta$ are the integral and differential terms. The notation $Q \succ 0$ denotes positive definiteness. $x$, $\hat{x}$ and $\tilde{x}$ denote the state, estimated state and state error; $u$ and $\tilde{u}$ denote the inputs for tracking and regulating problems, respectively. Subscript $i$ indicates matrix/vector component and $k$ the prediction step, superscript $i$ is the critical region index. $I_m$ is an identity matrix of order $m$. 

\begin{figure}[t]
\centering
\includegraphics[width=3.2in]{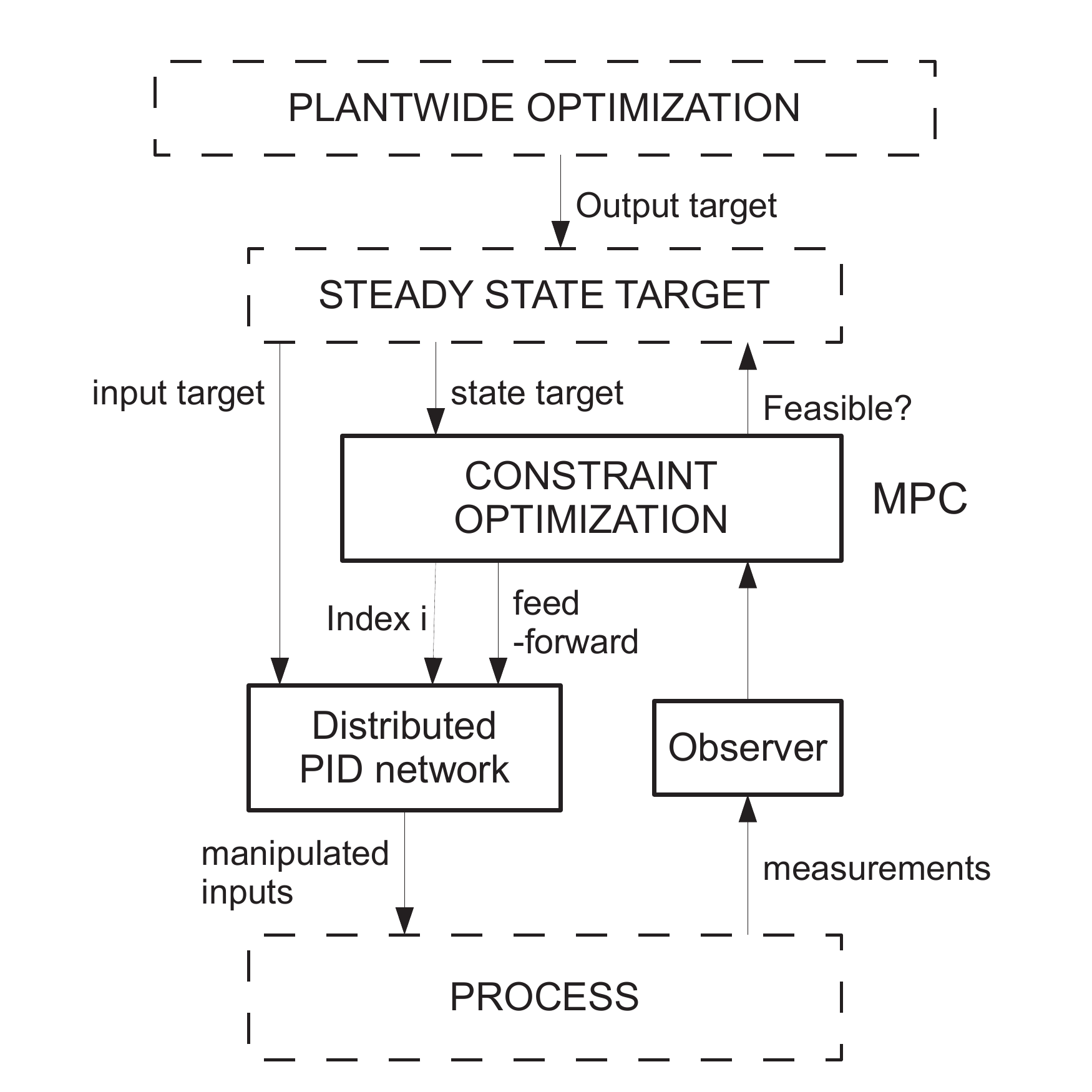}
\caption{Optimization control with multi-layers.}
\label{fig1}
\end{figure}
\section{Preliminaries}\label{preli}
To obtain a linear feedback involving proportional-integral-differential gains, it is necessary to form a system state that contains the corresponding variables. Provided that is the case, an optimal linear feedback gain is also an optimal PID gain. This section introduces the augmented PI/PID-state systems and covers the analysis of their controllability and observability. 

\subsection{Plant Model}
Consider a linear time-invariant system 
\begin{IEEEeqnarray}{rCl}
x(k+1)&=& Ax(k) + Bu(k)\nonumber \\
v(k)&=& C_vx(k)\nonumber \\
y(k) &=& Cx(k).\label{mdl}
\end{IEEEeqnarray}
subject to the constraint
\begin{equation}
Ex(k) + Fu(k) \leq G.\label{cons}
\end{equation}
In \eqref{mdl}, $x(k)\in \mathbb{R}^{n}, u(k)\in \mathbb{R}^{m}, v(k)\in \mathbb{R}^{q}\, (q\leq n)$,  and $y(k)\in \mathbb{R}^{p}$ are the state, input, tracked output and measured output. Assume $(A,B)$ is controllable and $(A,C)$ is observable; $C,C_v$ having full row rank; $E, F, G$ are appropriate matrices defining the state and input constraints. 

The plant model \eqref{mdl} is augmented with an integral of the tracked output $\sum{v(k)}$ to ensure zero offset during the steady state. The following PI-state model is used
\begin{IEEEeqnarray}{rCl}
\begin{bmatrix}x(k+1)\\ \sum{v(k+1)}\end{bmatrix}&=&
\begin{bmatrix}A& 0\\C_v& I_q\end{bmatrix}
\begin{bmatrix}x(k)\\ \sum{v(k)}\end{bmatrix}+
\begin{bmatrix}B\\ 0\end{bmatrix}u(k)\nonumber \\
y(k) &=& Cx(k).\label{augmdl}
\end{IEEEeqnarray}
In special cases, $C_v=I$ requires a full-state tracking while $C_v=C$ expects only output tracking.

\begin{proposition}
The PI-augmented system \eqref{augmdl} is detectable. Furthermore, it is controllable if and only if (A,B) is controllable and
\begin{equation}
rank \begin{bmatrix}A-I_n& B\\ C_v& 0\end{bmatrix}=n+q
\label{hauCtrl}
\end{equation}
\end{proposition}

\begin{proof}
The Hautus condition for observability is
\begin{equation}
rank\begin{bmatrix}A^T-\lambda I_n& C_v^T& C^T\\ 0& I_q-\lambda I_q& 0 \end{bmatrix}=n+q\text{\quad for all }\lambda\in \mathbb{C}.
\label{hauObsv}
\end{equation}

The condition \eqref{hauObsv} does not hold only at $\lambda=(1,0)$, but the unobservable integrating state can be controlled to decay to a constant so the system is detectable.

Similarly, \eqref{hauCtrl} follows directly from Hautus controllability where only the case of $\lambda=(1,0)$ is to check.
\end{proof}

By addition of the differential term, the PID-state system presents as
\begin{IEEEeqnarray}{rCl}
\begin{bmatrix}x(k+1)\\ \sum{v(k+1)}\\ \Delta v(k+1)\end{bmatrix}&=&
\begin{bmatrix}A& 0& 0\\ C_v& I_q& 0\\ C_v(A-I_n)& 0& 0\end{bmatrix}
\begin{bmatrix}x(k)\\ \sum{v(k)}\\ \Delta v(k)\end{bmatrix}+\begin{bmatrix}B\\ 0\\ C_vB\end{bmatrix}u(k)\nonumber \\
y(k) &=& Cx(k).\label{augmdl2}
\end{IEEEeqnarray}
This PID-augmented system is detectable and stabilizable. The proof is similar to Proposition \ref{hauCtrl}.

\begin{rmk}The number of tracked variables is presumed less than or equal to the number of manipulated variables ($q\leq m$ for PI case and $q\leq m/2$ for PID case); the other case was well treated in \citet{Mae09Linear}.
\end{rmk}

The objective is to design a finite-horizon optimal control based on the augmented system \eqref{augmdl} or \eqref{augmdl2} so that $v(k)$ tracks a piece-wise constant reference.

\subsection{Observer Design}
From the system detectability, an observer can make use of the system \eqref{mdl} to estimate the current state, and simply calculate the integral and differential state through a sum of the estimated $\hat{v}(k)=C_v\hat{x}(k)$ and its difference.

Since $(A,C)$ is observable, the observer is designed as
\begin{IEEEeqnarray}{rCl}
\hat{x}(k)&=&A\hat{x}(k-1)+Bu(k-1)\nonumber\\
&&+L_x[-y(k-1)+C\hat{x}(k-1)]\nonumber\\
\sum{\hat{v}(k)}&=& \sum{\hat{v}(k-1)}+C_v\hat{x}(k-1)\nonumber\\
&&+C_vL_x[-y(k-1)+C\hat{x}(k-1)]\nonumber\\
\Delta \hat{v}(k) &=& C_v(\hat{x}(k)-\hat{x}(k-1))
\end{IEEEeqnarray}
It is only necessary to design the observer gain $L_x$ as $eig (A+L_xC)<1$ so that $\hat{x}(k)-x(k)\rightarrow 0$. This automatically leads to $\Delta \hat{v}(k)$ being stable. The integral estimation error is not required to decay to zero, but a steady state because $\sum{\hat{v}(k)}-\sum{v(k)}\rightarrow const$ means $\hat{v}(k)-v(k)\rightarrow 0$.

\section{Controller Design}\label{MPCPID}
\subsection{MPC tracking structure}\label{MPC}
This section will outline the general MPC controller design for a state space model that results in PI/PID control implementation fulfilling the constraints.

Consider the linear system with constraints $z(k+1)=A_mz(k)+B_mu(k)$. Define the operating points $(z_s, u_s)$ and the deviation variables
\begin{IEEEeqnarray}{rCl}
\tilde{z}(k)&=&z_s-z(k)\nonumber \\
\tilde{u}(k)&=&u_s-u(k),
\label{dev}
\end{IEEEeqnarray}
\begin{IEEEeqnarray}{rCl}
\text{then\qquad}\tilde{z}(k+1)&=&A_m\tilde{z}(k)+B_m\tilde{u}(k)\nonumber\\
\tilde{y}(k) &=& C_m\tilde{z}(k).
\label{devmdl}
\end{IEEEeqnarray}

The finite-horizon quadratic optimal control problem is posed as 
\begin{IEEEeqnarray}{lCl}
V_N^o(\tilde{z}_0,\tilde{U}&&)=\underset{\tilde{U}}{\operatorname{min.}}\,\tilde{z}_N^TP\tilde{z}_N\nonumber\\
&&\qquad+\sum_{k=0}^{N-1}({\tilde{z}_k}^TC_m^TQC_m\tilde{z}_k+\tilde{u}_k^TR\tilde{u}_k)\label{MPClaw}\\
subj.\,to \, && \tilde{z}_k\in X, \tilde{u}_k\in U \quad \forall k\in {0,...,N-1},\nonumber\\
&&\tilde{z}_0\in X_0,\,\tilde{z}_N\in X_f,\nonumber \\
&& \tilde{z}_{k+1}= A_m\tilde{z}_k+ B_m\tilde{u}_k,\, \nonumber
\end{IEEEeqnarray}
where $\tilde{U}=\{\tilde{u}_0,...,\tilde{u}_{N-1}\}$. Here $Q\geq 0, R\succ 0$ are the weighting matrices, $(Q^{1/2},A_m)$ is detectable; $P\geq 0$ is the terminal penalty matrix. $X_0, X_f$ are the initial feasible set and the terminal constraint set. Note that $X, U$ are translated constraints from \eqref{cons} through the transformation in \eqref{dev}. By the receding horizon policy, only $\tilde{u}_0$ is applied to the plant.

\begin{assumption}
The state and input constraints are not active for $k\geq N$. Also, $X_f$ contains the origin.
\end{assumption}

The optimizer $\tilde{U}$ stabilizes \eqref{devmdl} if the value function $V_N^o(\tilde{z})$ corresponds to a local Lyapunov function $V_f$ within the terminal set $X_f$. In addition, the decay rate of that Lyapunov function must be larger than the stage cost \citep{May00Constrained}. Under this setup, any admissible $\tilde{z}_0$ is steered to a level set of $V_f$ (and so $X_f$) within N steps, after which convergence and stability of the origin follows. In other words, $z_k$ is stable at $z_s$ for $k\geq N$.

Therefore, given the state and input weighting matrices $Q,R$, one would want to first compute an unconstrained stabilizing feedback $\tilde{u}=K\tilde{z}$ and its Lyapunov function $V(\tilde{z})$ that satisfy
\begin{IEEEeqnarray}{rCl}
V_f(\tilde{z})&=& \tilde{z}^TP\tilde{z}\geq 0,\nonumber\\
\Delta V_f(\tilde{z})&=& \tilde{z}^TA_K^TPA_K\tilde{z}- \tilde{z}^TP\tilde{z}\nonumber \\
&\leq &-\tilde{z}^TQ\tilde{z}-\tilde{z}^TK^TRK\tilde{z},\, \forall \tilde{z}\in X_f,\label{conMPC}
\end{IEEEeqnarray}
where $A_K=A_m+B_mK$. The other ingredients of MPC formula are then determined as follows.
\begin{itemize}
	\item $X_f$ is the maximal positively invariant polyhedron of $\tilde{z}_{k+1}= A_m\tilde{z}_k + B\tilde{u}_k$ with respect to $\tilde{z}_k\in X, \tilde{u}_k\in U$. As commented in \citet{Rawlings09Model}, if $X_f$ is ellipsoidal, the problem is no longer a quadratic program but a convex program but can be solved with available softwares.
	\item $X_0$ is the N-step stabilizable set of the system \eqref{MPClaw} with respect to $X_f$. $N$ is a trade-off value between the complexity of MPC problem and a larger set $X_0$ (i.e. larger initial error $\tilde{z}_0$).
	\item $P$ is chosen as the solution of the equality in \eqref{conMPC}, the unique positive-definite solution of 
a discrete Lyapunov equation once $K$ is known \citep{Gri05Stabilizing}.
\end{itemize}
$X_0, X_f$ can be calculated analytically using the method detailed in \citet{Bla99Set,Ale06}. 

A popular choice for $K$ is obtained from the LQR gain with weighting matrices $Q, R$ \citep{Chm96constrained, Sco98Constrained}. However, in this note, it is left as a general stabilizing gain $K$ that will be computed in the next section.

\subsection{Computation of Stabilizing PI/PID}\label{PIDsec}
This session describes a method to compute an unconstrained feedback gain $K$ that is used to reconstruct the MPC formula \eqref{MPClaw}. It is because this gain would result in PI/PID controllers, as shown in the following theorem. For the general case, let $z=\begin{bmatrix}x^T& \sum{v^T}& \Delta v^T\end{bmatrix}^T$.

\begin{thm}\label{theo1}
A control law $\tilde{u}(k)=K\tilde{z}(k)$ implements PID control on the system state $x(k)$ which ensures robust tracking for $v(k)$.
\end{thm}
\begin{proof}
Because $\tilde{z}$ is the augmented state error, the control law is written as
\begin{IEEEeqnarray}{rCl}
\tilde{u}(k)&=& K\tilde{z}(k)\nonumber \\
&=& K_1\tilde{x}(k) + K_2\sum{\tilde{v}(k)}+ K_3\Delta\tilde{v}(k)\nonumber\\
&=& K_1\tilde{x}(k)+K_2C_v\sum{\tilde{x}(k)}+K_3C_v\Delta\tilde{x}(k).\IEEEeqnarraynumspace
\label{PID}
\end{IEEEeqnarray}
Since $rank(C_v)=q\leq n$, there are $m\times(n-q)$ P controllers and $m\times q$ PID controllers. In particular, PID control is applied to the state variables which influence the tracked output $v(k)$, so they are robust against disturbances.
\end{proof}

Let $(\bar{A}, \bar{B}, \bar{x})$ be the augmented model and state of \eqref{augmdl}.

\subsubsection{PI Controller ($K_3=0$)}
For this case, it is essential to obtain the feedback gain for $z = \bar{x}=\begin{bmatrix}x^T& \sum{v^T}\end{bmatrix}^T$. The PI control can be formulated by applying LQR to the error model of \eqref{augmdl}
to produce a PI control law $u(k)=K_{PI}z(k)$. From here simply take $(A_m,B_m,C_m)=(\bar{A},\bar{B},\bar{C}), K=K_{PI}$ and use \eqref{conMPC} to apply the MPC formula.

\subsubsection{PID Controller}
To get a non-trivial differential gain $K_3$, one can treat the differential term as an output feedback of the system \eqref{augmdl}. Define $\phi=\begin{bmatrix}x^T& \sum{v^T}& \phi_3^T\end{bmatrix}^T$ where $\phi_3 = \Delta v - C_vBu=C_v(Ax+Bu-x)-C_vBu = C_v(A-I_n)x$. Then
\begin{IEEEeqnarray}{rCl}
\bar{x}(k+1)&=& \bar{A}\bar{x}(k)+\bar{B}u(k)\nonumber \\
\phi(k) &=& \bar{C}\bar{x}(k)=\begin{bmatrix}I_n&0\\ 0& I_q\\ C_v(A-I_n)& 0\end{bmatrix}\bar{x}(k).
\label{sofmdl}
\end{IEEEeqnarray}

Design of static output feedback (SOF) $u(k)=F\phi(k)$ for the discrete time system above has been investigated in \citet{Gar03Robust,Bar05Static,Don07Static,He08Output} which use LMI conditions. There exists more outputs than inputs in this case, so we present a simple solution to determine $F$ in Theorem \ref{SOF} \citep{Bar05Static}. In that work, the solution can be extended to the $H_\infty$ design, but the detail is omitted here for simplicity (refer to Remark \ref{rmk3}).



\begin{thm}\label{SOF}
System \eqref{sofmdl} is stabilizable by a static output feedback if there exist a symmetric positive definite matrix $P_0\in\mathbb{R}^{(n+q)\times(n+q)}$ and a positive scalar $\sigma\in\mathbb{R}$ such that
\begin{equation}\label{LMI1}
\bar{A}^TP_0\bar{A}-P_0+\sigma \bar{B}\bar{B}^T\prec 0\\
\end{equation}
is satisfied. Furthermore, the SOF gain $F$ can be obtained by solving
\begin{equation}\label{LMI2}
(\bar{A}+\bar{B}F)^TP_0(\bar{A}+\bar{B}F)-P_0\prec 0.
\end{equation}
\end{thm}

Conditions \eqref{LMI1}, \eqref{LMI2} can be solved as two LMI problems. Once we have found a stabilizing output feedback $u(k)=F\phi(k)$ or equivalently $\tilde{u}(k)=F\tilde{\phi}(k)$, it can be rewritten in the PID form \cite{Zhe02design} as
\begin{IEEEeqnarray}{rCl}
\tilde{u}(k)&=&F_1\tilde{x}^T+F_2\sum{\tilde{v}^T}+F_3\Delta \tilde{v}+F_3C_vB\tilde{u}(k)\\
\text{so\ }\tilde{u}(k)&=&(I_n+F_3C_vB)^{-1}[F_1\tilde{x}(k)^T
+F_2\sum{\tilde{v}(k)}+F_3\Delta \tilde{v}(k)]\nonumber\\
&=&K_{PID}\tilde{z}(k),
\end{IEEEeqnarray}
where $\tilde{z}=\left[\tilde{x}^T\ \sum{\tilde{v}^T}\ \Delta \tilde{v}^T\right]^T$. The invertibility of matrix $(I_n+F_3C_vB)$ is a necessary condition to render $K_{PID}$.

The MPC formula takes $\tilde{z}(k+1) = A_m\tilde{z}(k)+B_m\tilde{u}(k)$,
\begin{IEEEeqnarray}{l}\label{Am}
A_m=\begin{bmatrix}A& 0& 0\\ C_v& I_n& 0\\ C_v(A-I_n)& 0& 0\end{bmatrix}, B_m=\begin{bmatrix}B\\ 0\\C_vB\end{bmatrix}u(k)\IEEEeqnarraynumspace
\end{IEEEeqnarray}
and $K=K_{PID}$ to apply into \eqref{conMPC}.

\begin{rmk} \label{rmk2} Applying LQR directly to PID state for system \eqref{Am} will not result in a PID controller. In fact, since $A_m$ is no longer full rank, the optimal input $\tilde{u}(k)=(R+B_m^TQB_m)^{-1}B_m^TQA_m\tilde{z}(k)$ depends only on the first two components of $\tilde{z}(k)$, so it is not a full PID but a PI gain. However, we realize that increasing the weight on $\Delta v$ of $Q$ does reduce the overshoot and enhance the disturbance response of $v(k)$.
\end{rmk}

\begin{rmk} \label{rmk3} The PID design for multivariable systems used in this paper is not unique. It is possible to use other techniques such as \citet{Dic09parameter,Soy03Fast,Tos09Robust} to derive a robust PID gain before applying it into MPC.
\end{rmk}

\section{From parametric MPC to PID gain scheduling controllers}\label{fastMPC}

The result from Section III holds when it is applied to either an \emph{online} or \emph{offline} MPC formulation. In this section, we particularly use parametric MPC (offline) to demonstrate the PID gain scheduling realization.

\subsection{Parametric MPC}
Observe that the problem \eqref{MPClaw} minimizes a convex value function subject to a convex constraint set. We have the following definition
\begin{definition}[Critical Region]
A critical region is defined as the set of parameters $\tilde{z}$ for which the same set of constraints is active at the optimum $(\tilde{z},\bar{U}^0(\tilde{z}))$.
\end{definition}

In other words, if the constraints in \eqref{MPClaw} is presented as $G\bar{U}\leq S\tilde{z} + W$ and $A$ is an associated set of row index,
\begin{IEEEeqnarray}{l}
CR_A =\{\tilde{z}\in X_0\,|\,G_i\bar{U}^0=S_i\tilde{z}+W_i\text{ for all } i\in A\}\IEEEeqnarraynumspace
\label{eq:}
\end{IEEEeqnarray}

In \citet{Bao02,Ton03algorithm}, it is shown that these critical regions are a finite number of closed, non-overlapped polyhedra and they covers completely $X_0$. Since $\tilde{U}=\{\tilde{u}_0,...,\tilde{u}_{N-1}\}$, the same properties apply for $\tilde{u}^0_0$. Theorem \ref{theo2} states the key result (see \citet{Bem02explicit}).

\begin{thm}[Parametric solution of MPC]
\label{theo2}
The optimal control law $\tilde{u}^0_0=f(\tilde{z}_0),f:X_0\mapsto U$, obtained as a solution of \eqref{MPClaw} is continuous and piecewise affine on the polyhedra
\begin{equation}
f(\tilde{z})=F^i\tilde{z}+g^i\quad if\, \tilde{z}\in CR^i, i=1,...,N^r,
\label{mulctrl}
\end{equation}
where the polyhedral sets $CR^i\triangleq \{H^i\tilde{z}\leq k^i\}, i=1,...,N^r$ are a partition of the feasible set $X_0$.
\end{thm}

\subsection*{Tracking for piecewise constant setpoint}
Recall the admissible set $X_0$ the MPC controller can stabilize depends on the linearized model $x(k+1)=f(x(t))|_{x=x_s}$ and control horizon $N$. Tracking of a new setpoint can be done by increasing $N_2$ based on the new model so that a jump in reference $z_{s1}\rightarrow z_{s2}$ is feasible within $N_2$ steps. 

\begin{figure}[!t]
\centering
\includegraphics[width=3.0in]{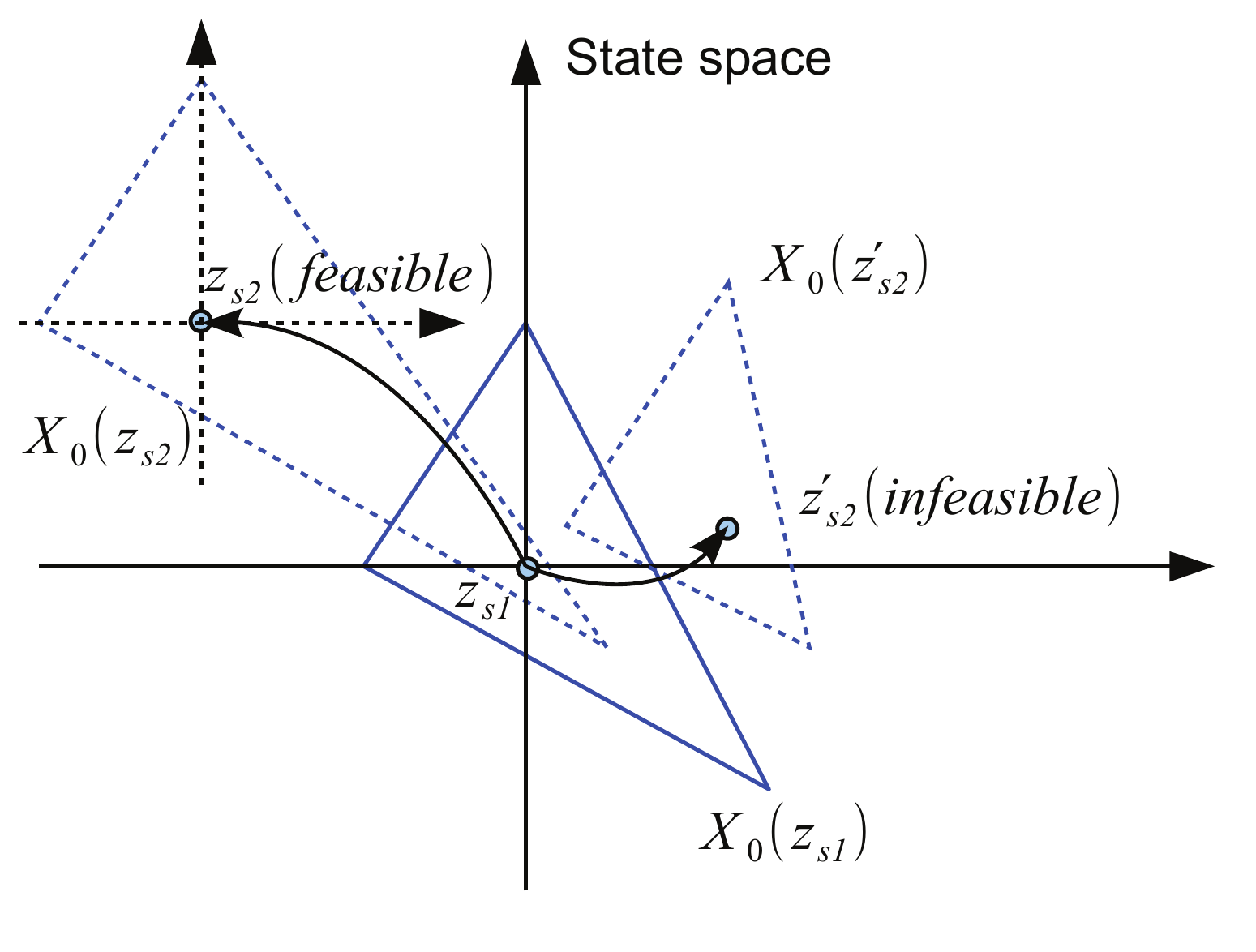}
\caption{Feasibility check of new setpoint $z_{s2}$.}
\label{feas}
\end{figure}

In the case of fixed $N$, Corollary \ref{reffsb} states the necessary and sufficient condition for a new feasible setpoint

\begin{corollary}\label{reffsb}
With a fixed-horizon proposed controller, a change in setpoint $z_{s1}\rightarrow z_{s2}$ is feasible if and only if $z_{s1}-z_{s2}\in X_0(z_{s2})$.
\end{corollary}
\begin{proof}
The proof can be inferred from Fig \ref{feas}. If $z_{s1}$ is out of the maximal admissible region $X_0(z_{s2})$ constructed around $z_{s2}$, it is impossible to drive the current error $\tilde{z}=z_{s1}-z_{s2}$ to zero with the existing controller.
\end{proof}

Corollary \ref{reffsb} suggests a way to detect if a new setpoint is feasible so that the local optimization for steady state target can recalculate $z_s$ early before the infeasibility happens. One can use a single model and treat the model mismatch at a different operating point as disturbance, but generally $X_0$ still needs to be rebuilt through \eqref{dev} because the constraints change with setpoint relocation.


\subsection{PID Gain Scheduling Design}
The optimal input of MPC is applied for regions outside $X_f$. When $\tilde{z}(k)$ reaches $X_f$, the system will be stabilized by the pure gain $F^0=K$. Therefore, one practical way to design PID for constrained systems is designing a PID gain for its unconstrained region, which has been accomplished in Section \ref{MPCPID}, and applying these settings on the MPC formulation \eqref{MPClaw}.

\begin{figure}[t]
\centering
\includegraphics[width=3.6in]{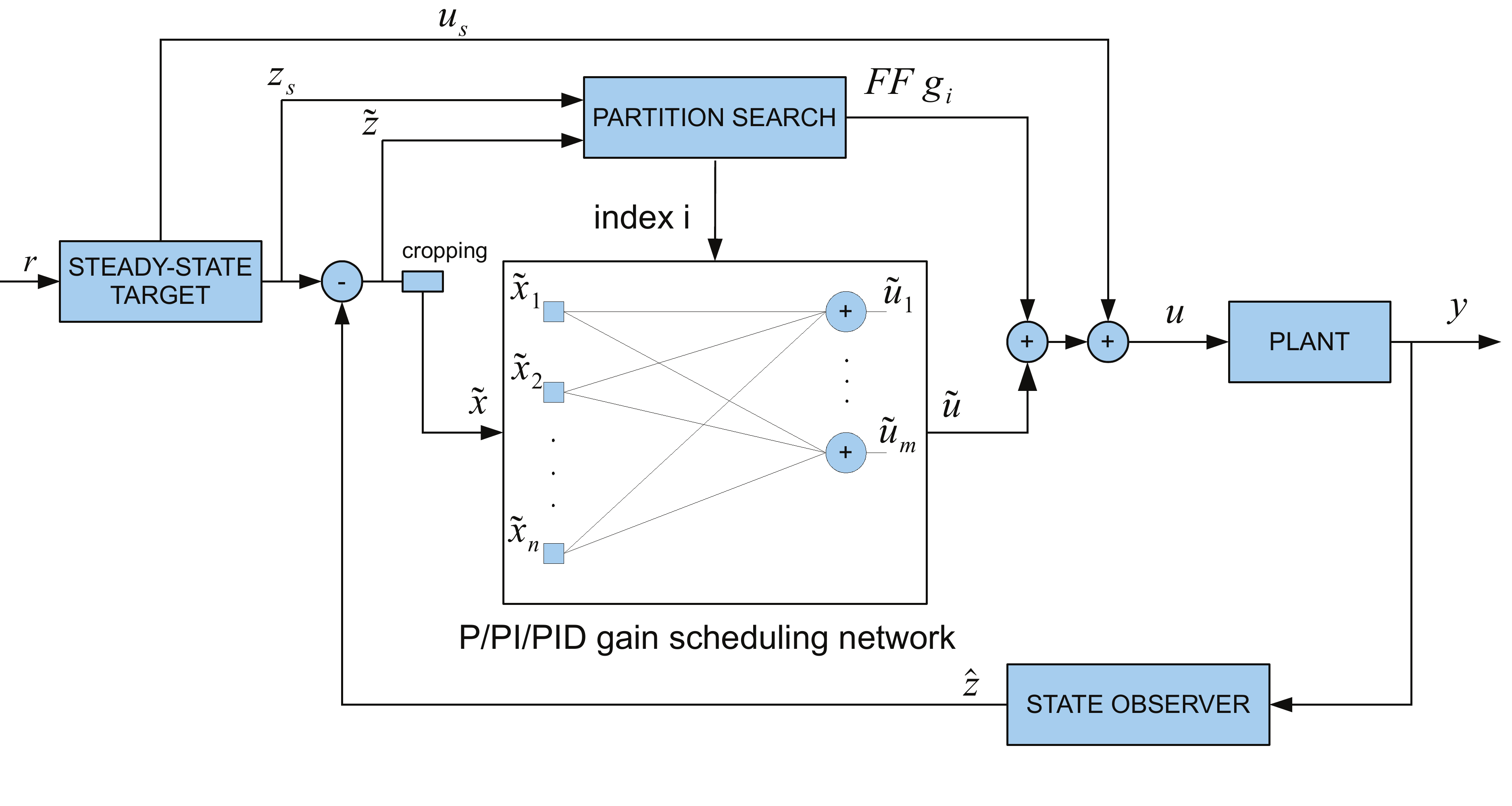}
\caption{Proposed PID gain scheduling structure.}
\label{ctrlDiag3}
\end{figure}

Fig. \ref{ctrlDiag3} shows a series of PIDs plus a single feedforward vector where the controller gains are determined from \eqref{mulctrl}. Each of the PIDs is fully flexible (might contain only P or PI components) and have its own look-up gain scheduling for different partition indexes. At each time step, the proposed scheme would look for the region in which the augmented error $\tilde{z}(k)$ lies in. This search engine would broadcast the region index $i$ to the PID network. The feedforward term associated with region $i$ is added to compensate the active constraints. Non-zero tracking accounts for the addition of steady state input and recovers the input delivered to the plant.

\begin{rmk}
As seen from Fig. \ref{ctrlDiag3}, the PID network consists of one-to-one mappings between each state variable error of the original state $x$ and an input. This fact results from equation \eqref{PID}.
\end{rmk}


\section{Example}


\begin{figure}[t]
\centering
\subfigure{
\label{x1}
\includegraphics[width=\columnwidth]{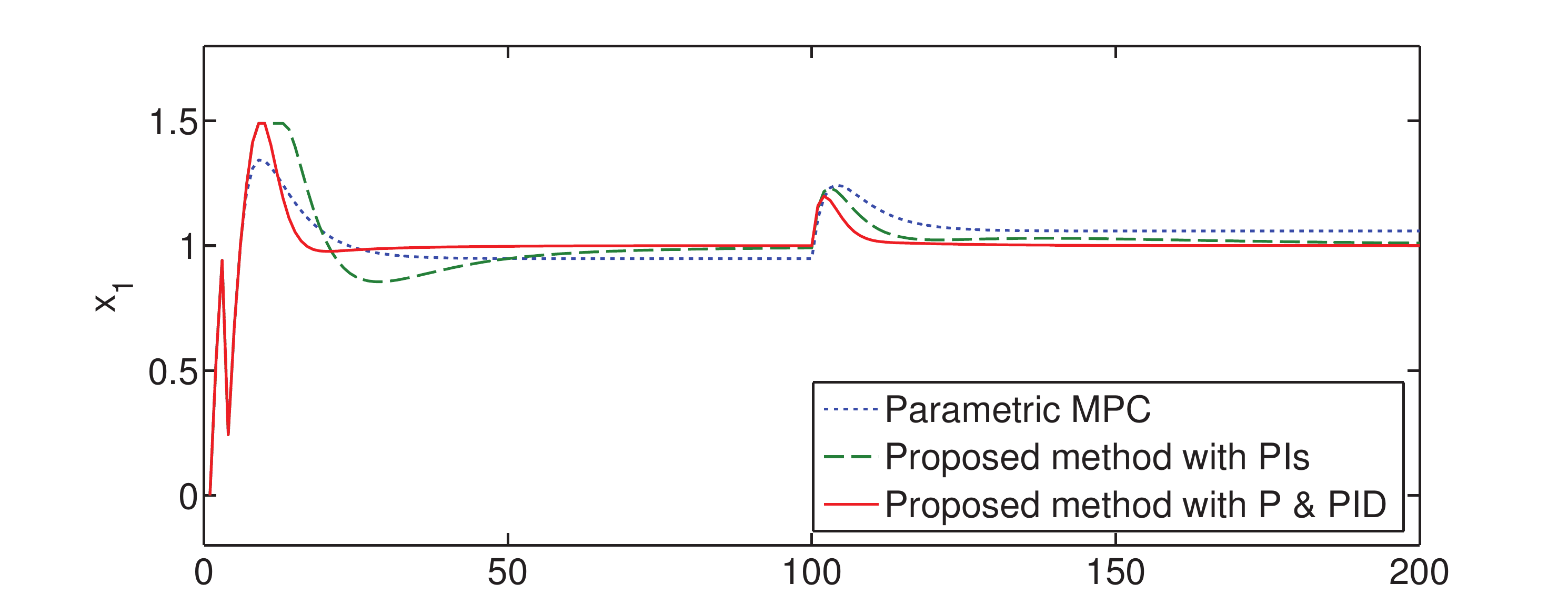}}
\subfigure{
\label{u1}
\includegraphics[width=\columnwidth]{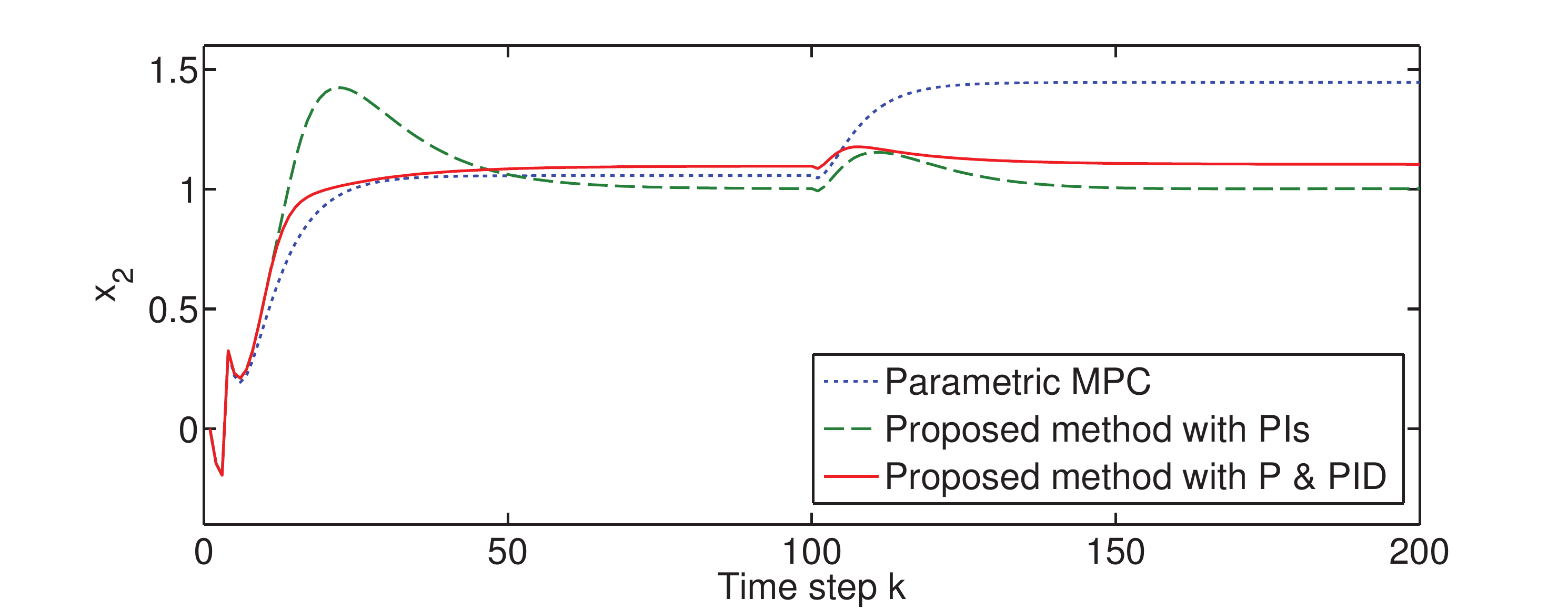}}
\subfigure{
\label{x2}
\includegraphics[width=\columnwidth]{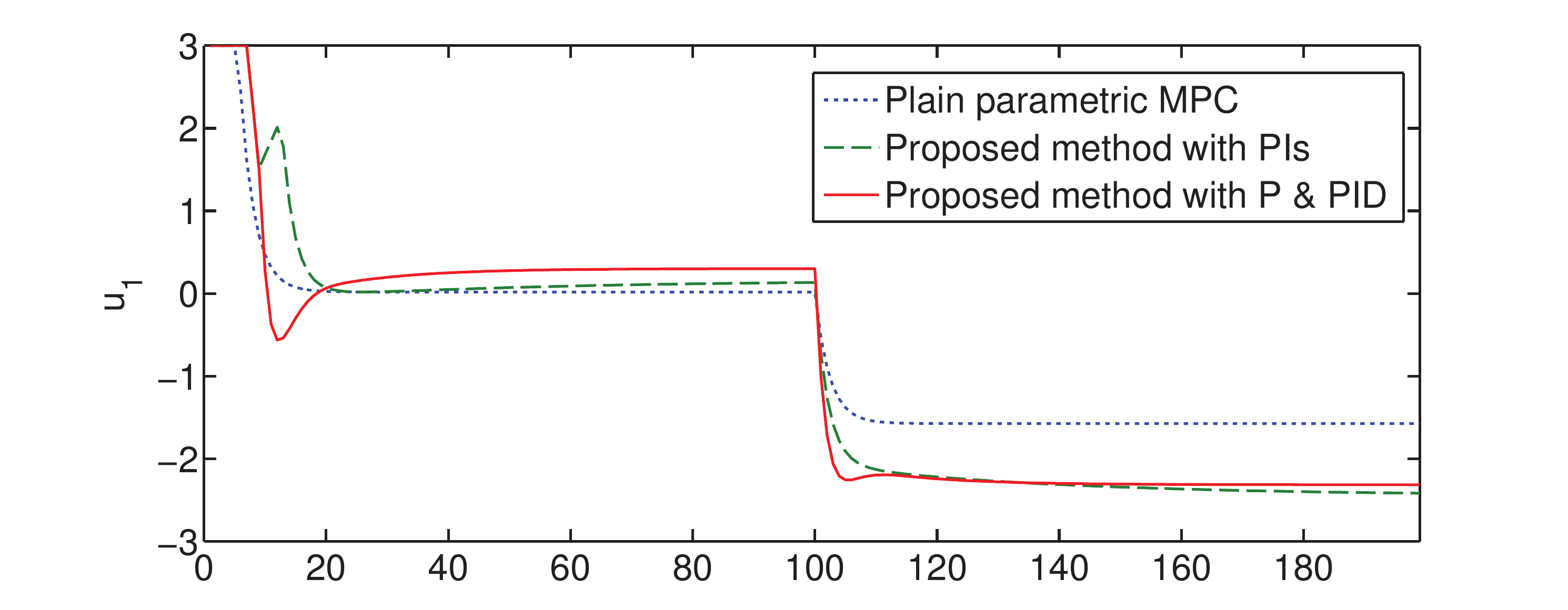}}
\subfigure{
\label{u2}
\includegraphics[width=\columnwidth]{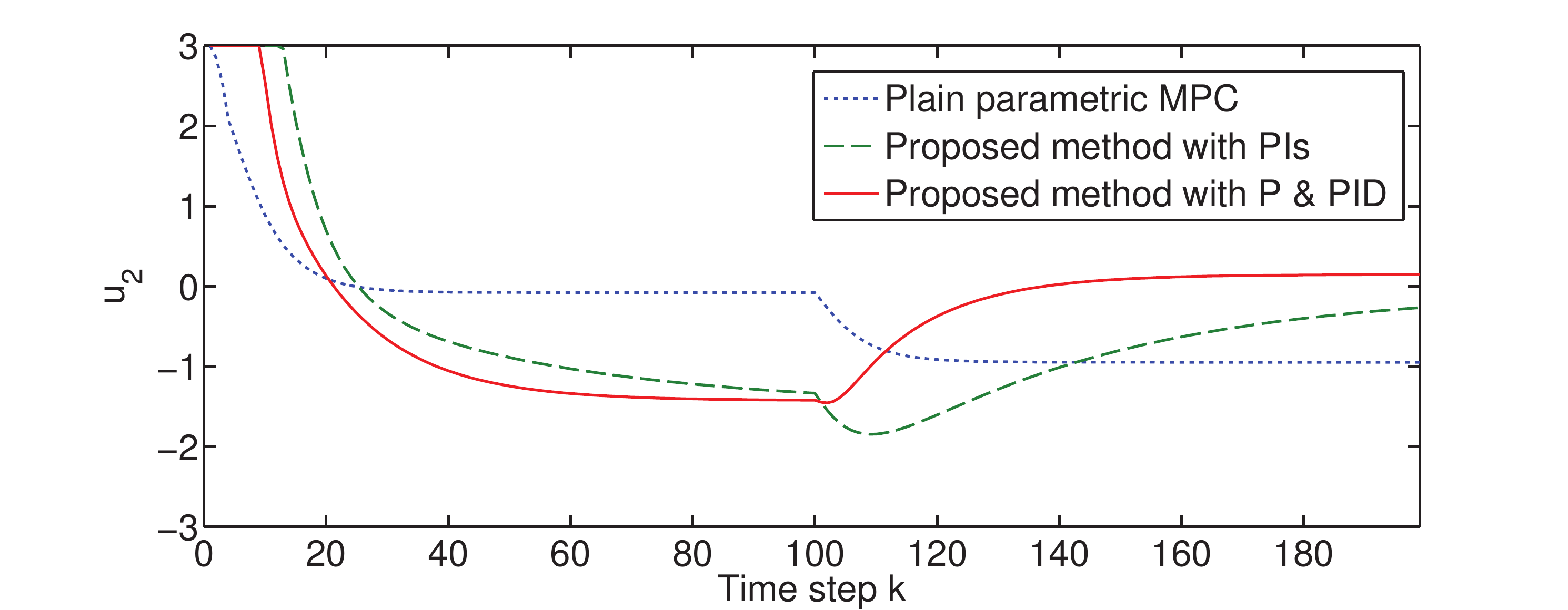}}
\caption{State responses and control inputs under disturbances at transient and steady-state.}
\label{stateinput}
\end{figure}

The proposed control design was illustrated in the following example, generalized from \citet{Bem02explicit} with two inputs. Consider a continuous stirred-tank reactor model
\begin{IEEEeqnarray}{rCl}
A&=&\begin{bmatrix}0.7326& -0.0861\\0.1722& 0.9909\end{bmatrix}, B=\begin{bmatrix}0.0609& 0\\ 0& 0.0064\end{bmatrix},\nonumber\\
C&=&\begin{bmatrix}1& 0\\ 0& 1\end{bmatrix},\, X=\left\{x\in \mathbb{R}^2|\begin{bmatrix}-0.5\\ -0.5\end{bmatrix}\leq x \leq \begin{bmatrix}1.5\\ 2.5\end{bmatrix}\right\},\nonumber\\
U&=&\left\{u\in \mathbb{R} | \begin{bmatrix}-2\\ -2\end{bmatrix}\leq u \leq \begin{bmatrix}2\\ 2\end{bmatrix} \right\}.
\end{IEEEeqnarray}
The task was to track the level 1 $x_1$ with the reference $x_{1s} = 1$. To observe the robustness of tested controllers, the disturbances $d_1=[1;-0.5]$ (impulse), $d_1'=[0.01;-0.01]$ (additive) within an active constrained region at $k=3$ and  $d_2=[-0.15;0]$ (additive) at steady state $k=100$ were introduced.

The three following controllers were compared: simple parametric MPC (I), the whole state tracking with full PI (II) and $x_1$-tracking with partial PID (III). The prediction horizon (also control horizon in this case) is chosen as $N=2$. 

Tuning weighting matrices for PID control had been discussed in \citet{Ngu11Enhanced}. For PI, $z=\left[x_1^T\ x_2^T\ \sum{x_1^T}\ \sum{x_2^T}\right]^T$ and $Q=diag(1,1,0.001,0.001)$, $R=0.01I_2$; for PID $z=\left[x_1^T\ x_2^T\ \sum{x_1^T}\ \Delta x_1^T\right]^T$, $Q=diag(1,1,0.001,0.1)$, $R=0.01I_2$. MATLAB LMI solver was used to obtain the unconstrained PID gain for case III, and Multiparametric toolbox \citep{mpt} was applied to obtain the gains under critical regions.

The unconstrained gain $K$ in the three cases were 
\begin{IEEEeqnarray}{lCl}
K_I &=& \begin{bmatrix}4.501& 3.792\\0.711& 2.160\end{bmatrix},\nonumber\\
K_{II} &=& \begin{bmatrix}5.792& 6.353& 0.289& 0.469 \\1.103& 7.094& -0.579& 0.326\end{bmatrix},\nonumber \\
K_{III} &=&  \begin{bmatrix}0.493& 1.399& 0.139& -0.392\\3.014&  18.545& 1.765& -0.766 \end{bmatrix},\IEEEeqnarraynumspace
\end{IEEEeqnarray}
and they resulted in control laws with $8,\,14,\,12$ critical regions, respectively.

\begin{figure}[!ht]
\centering
\subfigure[]{\label{state1}
\includegraphics[width=\columnwidth,height=2.0in]{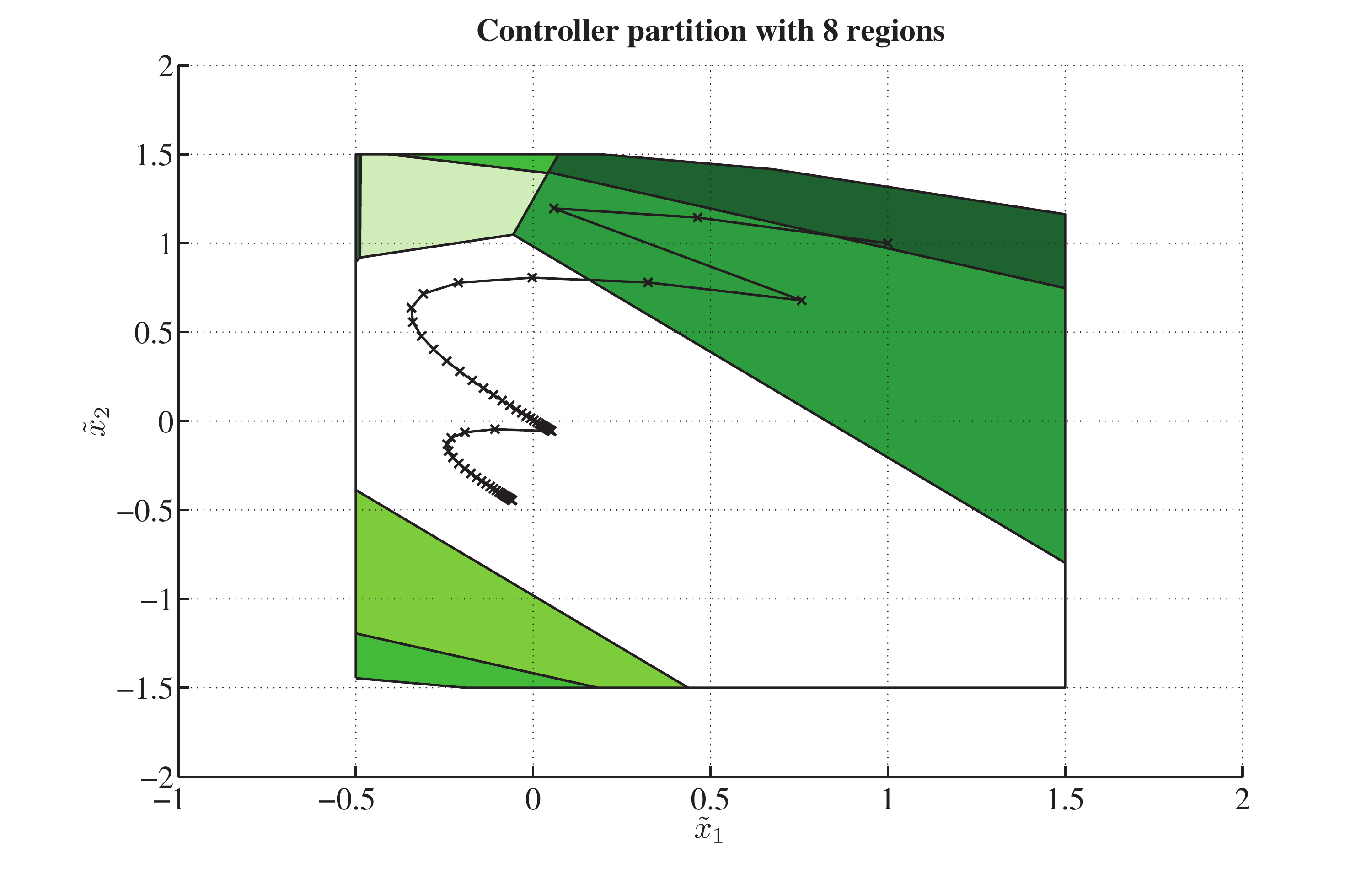}}
\subfigure[]{\label{state2}
\includegraphics[width=\columnwidth,height=2.0in]{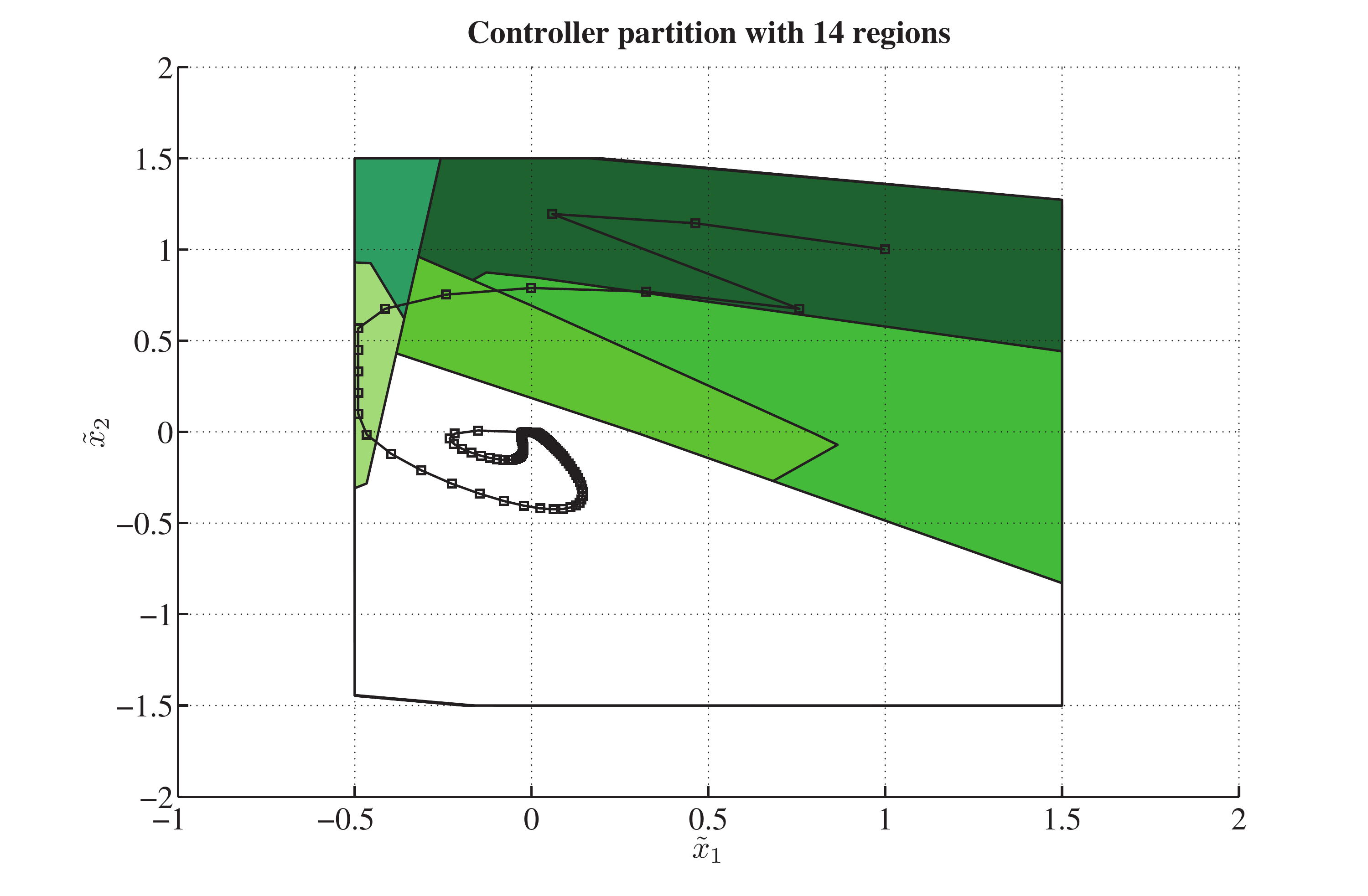}}
\subfigure[]{
\label{state3}
\includegraphics[width=\columnwidth,height=2.0in]{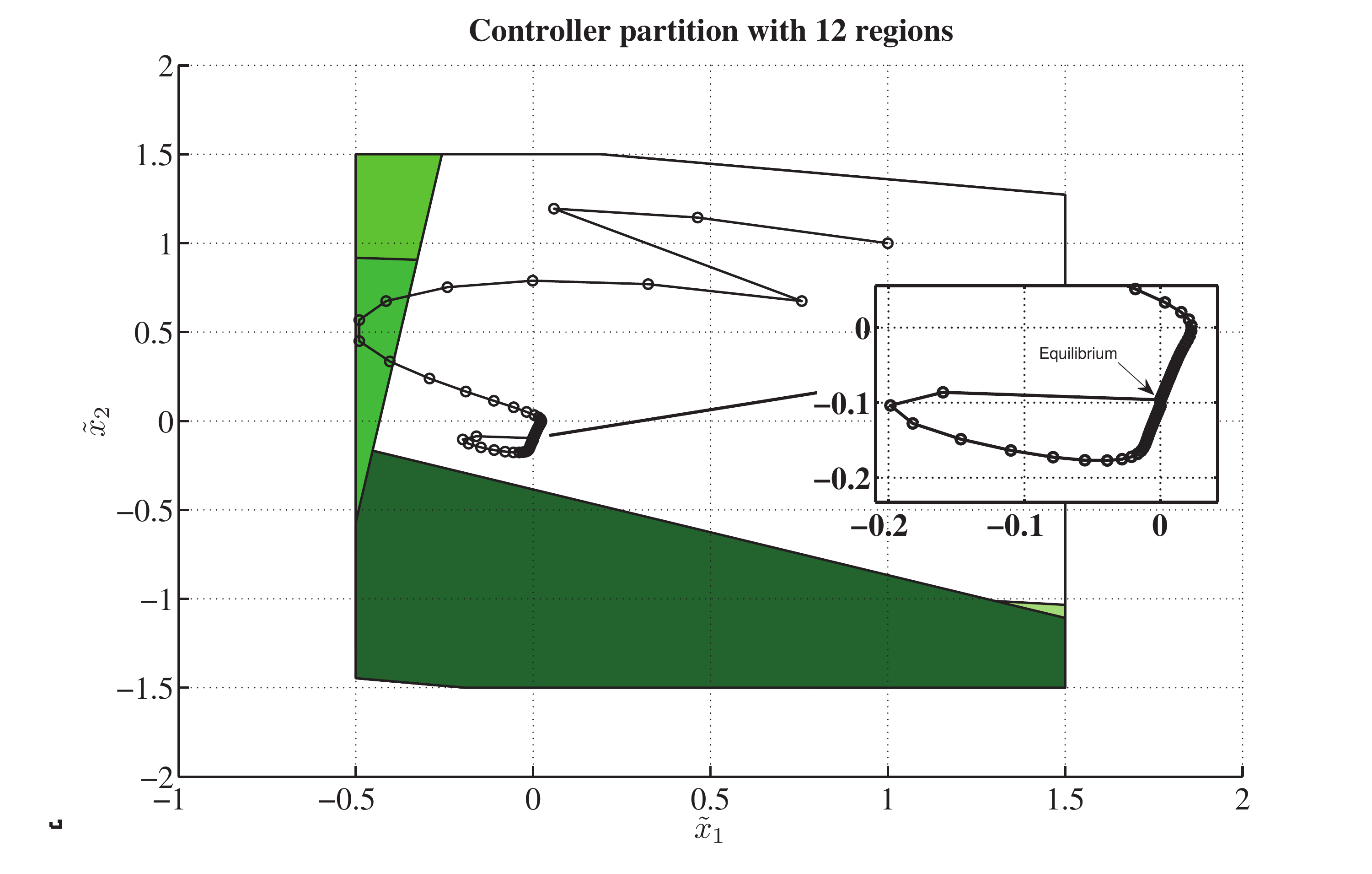}}
\caption{Controller partitions projected on subspace $[\tilde{x}_1,\tilde{x}_2]$ and the state trajectory with (a) Scheme I, (b) Scheme II (cut at $\sum{\tilde{x}_1}=\sum{\tilde{x}_2}=0$) and (c) Scheme III (cut at $\sum{\tilde{x}_1}=\Delta \tilde{x}_1=0$).}
\label{state}
\end{figure}

From the state response in Fig. \ref{stateinput}, we saw that the scheme I could not negate the additive disturbance happened either at an active constraint region or at steady state. It resulted in offset $\tilde{x}=\begin{bmatrix}0.06& -0.06\end{bmatrix}^T$ and $\tilde{x}=\begin{bmatrix}-0.05& -0.45\end{bmatrix}^T$, respectively. The scheme II could track both the state variables but with significant overshoot due to the regulation of $\sum{\tilde{v}}$ back to 0. That effect could be removed by tracking it to a constant (as a tuning parameter), but ignored in this example for simplicity. The scheme III tracked $x_1$ as required, and successfully forced the disturbance effect into $x_2$. The tracking under setpoint change and disturbance rejection also happened faster than scheme II. We stressed that all the three schemes were able to deal with the state and input constraints $x_1\leq 1.5$, $u\leq 3$ during transient stage because of the feedforward term $g^i$ in the parametric MPC law.

Fig. \ref{state} gave another perspective of the result. Provided that the impulse disturbance did not excite the current state out of the feasible region $X_0$, it was feasible to find an optimal input for all the three schemes. Secondly, scheme II hit on the outer constraint $\tilde{x}_1=-0.5\, (x_1=1.5)$ and took a long time to recover. Indeed, an integral windup happened at this upper output bound. The scheme III showed the full PID potential. It is known that the proportional-integral deals with the present and past behavior of the plant. The differential term predicts the plant behavior and can be used to stabilize the plant faster. This was in line with Remark \ref{rmk2}. The trajectory quickly returned to the origin in both cases of setpoint change and additive disturbance. Lastly, while scheme II regulated the state error back to the origin, scheme III only drove it to the axis $\tilde{x}_1=0$ as expected. It meant only $m\times q$ PIDs and $m\times (n-q)$ Ps were needed to track $q$ outputs.

In conclusion, it is observed that as long as the disturbance does not drive the equilibrium outside of the unconstrained region, output tracking using the integral state variables remains feasible. The robust stability during transient stage is inherent through the PID form. The robust stability around setpoint only concerns the PID control design described in Section \ref{PIDsec}, which can be improved further by $H_\infty$ approaches as stated in Remark \ref{rmk3}. Overall, extension to integral and differential terms is the natural to perform tracking control.

\section{Conclusion and Future work}

As it was never emphasized enough, the link between MPC and a robust linear controller at equilibrium is revisited in this paper. We modify the linear controller to be capable of offset-free tracking. The resultant control architecture is a PID gain scheduling network with a feedforward part to deal with state and input constraints. A simple test for setpoint tracking feasibility is also discussed. Finally, the example results show that the robustness stability of the proposed method is inherent within the PI/PID structure when disturbances arrives.


\bibliographystyle{elsarticle-harv}
\bibliography{D:/Dropbox/Research/referencelist}

\begin{thebibliography}{34}
\expandafter\ifx\csname natexlab\endcsname\relax\def\natexlab#1{#1}\fi
\expandafter\ifx\csname url\endcsname\relax
  \def\url#1{\texttt{#1}}\fi
\expandafter\ifx\csname urlprefix\endcsname\relax\def\urlprefix{URL }\fi

\bibitem[{Alessio et~al.(2006)Alessio, Bemporad, Lazar, and Heemels}]{Ale06}
Alessio, A., Bemporad, A., Lazar, M., Heemels, W., dec. 2006. Convex polyhedral
  invariant sets for closed-loop linear mpc systems. In: Decision and Control,
  2006 45th IEEE Conference on. pp. 4532 --4537.

\bibitem[{Alvarado et~al.(2008)Alvarado, Limon, Alamo, and Camacho}]{Alv08}
Alvarado, I., Limon, D., Alamo, T., Camacho, E., dec. 2008. Output feedback
  robust tube based mpc for tracking of piece-wise constant references. In:
  Decision and Control, 2007 46th IEEE Conference on. pp. 2175 --2180.

\bibitem[{Arousi et~al.(2008)Arousi, Schmitz, Bars, and Haber}]{Aro08}
Arousi, F., Schmitz, U., Bars, R., Haber, R., 2008. Robust predictive pi
  controller based on first-order dead time model. In: IFAC World Congress.

\bibitem[{Baotic(2002)}]{Bao02}
Baotic, M., Apr. 2002. An efficient algorithm for multiparametric quadratic
  programming. Tech. rep., ETH.

\bibitem[{Bara and Boutayeb(2005)}]{Bar05Static}
Bara, G., Boutayeb, M., feb. 2005. Static output feedback stabilization with h
  infinity performance for linear discrete-time systems. Automatic Control,
  IEEE Transactions on 50~(2), 250 -- 254.

\bibitem[{Bemporad et~al.(2003)Bemporad, Borrelli, and Morari}]{Bem03Min}
Bemporad, A., Borrelli, F., Morari, M., sept. 2003. Min-max control of
  constrained uncertain discrete-time linear systems. Automatic Control, IEEE
  Transactions on 48~(9), 1600 -- 1606.

\bibitem[{Bemporad et~al.(2002)Bemporad, Morari, Dua, and
  Pistikopoulos}]{Bem02explicit}
Bemporad, A., Morari, M., Dua, V., Pistikopoulos, E., 2002. The explicit linear
  quadratic regulator for constrained systems. Automatica 38(1), 3 -- 20.

\bibitem[{Blanchini(1999)}]{Bla99Set}
Blanchini, F., 1999. Set invariance in control. Automatica 35~(11), 1747 --
  1767.

\bibitem[{Camacho et~al.(2003)Camacho, Bordons, and
  Normey-Rico}]{Camacho03Model}
Camacho, E., Bordons, C., Normey-Rico, J., 2003. Model predictive control.
  Vol.~13. Springer Verlag.

\bibitem[{Chmielewski and Manousiouthakis(1996)}]{Chm96constrained}
Chmielewski, D., Manousiouthakis, V., 1996. On constrained infinite-time linear
  quadratic optimal control. Systems and Control Letters 29(3), 121 -- 129.

\bibitem[{Di~Cairano and Bemporad(2010)}]{Di10Model}
Di~Cairano, S., Bemporad, A., jan. 2010. Model predictive control tuning by
  controller matching. Automatic Control, IEEE Transactions on 55~(1), 185
  --190.

\bibitem[{Dickinson and Shenton(2009)}]{Dic09parameter}
Dickinson, P., Shenton, A., 2009. A parameter space approach to constrained
  variance pid controller design. Automatica 45~(3), 830 -- 835.

\bibitem[{Dong and Yang(2007)}]{Don07Static}
Dong, J., Yang, G., oct. 2007. Static output feedback control synthesis for
  linear systems with time-invariant parametric uncertainties. Automatic
  Control, IEEE Transactions on 52~(10), 1930 --1936.

\bibitem[{Froisy(2006)}]{Fro06Model}
Froisy, J., 2006. Model predictive control - building a bridge between theory
  and practice. Computers and Chemical Engineering 30, 1426--1435.

\bibitem[{Garcia et~al.(2003)Garcia, Pradin, Tarbouriech, and
  Zeng}]{Gar03Robust}
Garcia, G., Pradin, B., Tarbouriech, S., Zeng, F., 2003. Robust stabilization
  and guaranteed cost control for discrete-time linear systems by static output
  feedback. Automatica 39~(9), 1635 -- 1641.

\bibitem[{Grieder et~al.(2005)Grieder, Kvasnica, Baoti, and
  Morari}]{Gri05Stabilizing}
Grieder, P., Kvasnica, M., Baoti, M., Morari, M., 2005. Stabilizing low
  complexity feedback control of constrained piecewise affine systems.
  Automatica 41~(10), 1683 -- 1694.

\bibitem[{Grieder and Morari(2003)}]{Gri03}
Grieder, P., Morari, M., dec. 2003. Complexity reduction of receding horizon
  control. In: Decision and Control, 2003. Proceedings. 42nd IEEE Conference
  on. Vol.~3. pp. 3179 -- 3190 Vol.3.

\bibitem[{He et~al.(2008)He, Wu, G.P., and She}]{He08Output}
He, Y., Wu, M., G.P., L., She, J., nov. 2008. Output feedback stabilization for
  a discrete-time system with a time-varying delay. Automatic Control, IEEE
  Transactions on 53~(10), 2372 --2377.

\bibitem[{Kvasnica et~al.(2004)Kvasnica, Grieder, and Baoti}]{mpt}
Kvasnica, M., Grieder, P., Baoti, M., 2004. Multi-parametric toolbox (mpt).

\bibitem[{Maeder et~al.(2009)Maeder, Borrelli, and Morari}]{Mae09Linear}
Maeder, U., Borrelli, F., Morari, M., 2009. Linear offset-free model predictive
  control. Automatica 45~(10), 2214 -- 2222.

\bibitem[{Mark et~al.(2011)Mark, Kouvaritakis, Rakovic, and
  Cheng}]{Mar11Stochastic}
Mark, C., Kouvaritakis, B., Rakovic, S., Cheng, Q., 2011. Stochastic tubes in
  model predictive control with probabilistic constraints. IEEE Transactions on
  Automatic Control 56, 194--199.

\bibitem[{Mayne et~al.(2000)Mayne, Rawlings, Rao, and
  Scokaert}]{May00Constrained}
Mayne, D., Rawlings, J., Rao, C., Scokaert, P., 2000. Constrained model
  predictive control: Stability and optimality. Automatica 36, 789--814.

\bibitem[{Moradi(2003)}]{Mor03}
Moradi, M., june 2003. State space representation of mimo predictive pid
  controller. In: Control Applications, 2003. CCA 2003. Proceedings of 2003
  IEEE Conference on. Vol.~1. pp. 452 -- 457 vol.1.

\bibitem[{Nagy and Braatz(2004)}]{Nag04Open}
Nagy, Z., Braatz, R., 2004. Open-loop and closed-loop robust optimal control of
  batch processes using distributional and worst-case analysis. Journal of
  Process Control 14~(4), 411 -- 422.

\bibitem[{Nguyen et~al.(2011)Nguyen, Tan, and Huang}]{Ngu11Enhanced}
Nguyen, H., Tan, K., Huang, S., 2011. Enhanced predictive ratio control of
  interacting systems. Journal of Process Control 21~(7), 1115 -- 1125.

\bibitem[{Qin and Badgwell(2003)}]{Qin03survey}
Qin, S., Badgwell, T., 2003. A survey of industrial model predictive control
  technology. Control Engineering Practice 11, 733--764.

\bibitem[{Rawlings and Mayne(2009)}]{Rawlings09Model}
Rawlings, J., Mayne, D., 2009. Model Predictive Control: Theory and Design. Nob
  Hill Publishing.

\bibitem[{Sato(2012)}]{Sat12}
Sato, T., 2012. Predictive control approaches for pid control design and its
  extension to multirate system. Advances in Industrial Control, 4, 553--595.

\bibitem[{Scokaert and Rawlings(1998)}]{Sco98Constrained}
Scokaert, P., Rawlings, J., aug 1998. Constrained linear quadratic regulation.
  Automatic Control, IEEE Transactions on 43~(8), 1163 --1169.

\bibitem[{Soylemez et~al.(2003)Soylemez, Munro, and Baki}]{Soy03Fast}
Soylemez, M., Munro, N., Baki, H., 2003. Fast calculation of stabilizing pid
  controllers. Automatica 39~(1), 121 -- 126.

\bibitem[{Tatjewski(2008)}]{Tat08Advanced}
Tatjewski, P., 2008. Advanced control and on-line process optimization in
  multilayer structures. Annual Reviews in Control 32~(1), 71 -- 85.

\bibitem[{Tondel et~al.(2003)Tondel, Johansen, and Bemporad}]{Ton03algorithm}
Tondel, P., Johansen, T.~A., Bemporad, A., 2003. An algorithm for
  multi-parametric quadratic programming and explicit mpc solutions. Automatica
  39~(3), 489 -- 497.

\bibitem[{Toscano and Lyonnet(2009)}]{Tos09Robust}
Toscano, R., Lyonnet, P., 2009. Robust pid controller tuning based on the
  heuristic kalman algorithm. Automatica 45~(9), 2099 -- 2106.

\bibitem[{Zheng et~al.(2002)Zheng, Wang, and Lee}]{Zhe02design}
Zheng, F., Wang, Q., Lee, T., 2002. On the design of multivariable pid
  controllers via lmi approach. Automatica 38~(3), 517 -- 526.

\end{thebibliography}
\end{document}